\documentclass[11pt]{amsart}
\def\isdraft{0}

\usepackage[dvipsnames]{xcolor}
\usepackage{bbm}
\usepackage{graphicx}
\usepackage{bbm}
\usepackage{amsaddr} 
\usepackage{mathtools}
\usepackage{amsmath,amssymb,amsfonts,dsfont}
\usepackage{prettyref} 
\usepackage[utf8]{inputenc}
\usepackage{enumitem}
\usepackage[hyphens]{url}
\usepackage{tikz-cd}
\usepackage{tikz}
\usetikzlibrary{positioning}
\usetikzlibrary{arrows}
\usepackage{bussproofs}
\usepackage[mathscr]{euscript}
\usepackage{hyperref}
\usepackage{amsthm}
\usepackage{theapa}
\usepackage{a4wide}
\usepackage[color=black,textcolor=white\if\isdraft0,disable\fi]{todonotes}
\usepackage{stackengine}
\usepackage{stmaryrd}
\usepackage{wrapfig}
\usepackage{pdfpages}
\usepackage{marginnote}
\usepackage{float}
\graphicspath{{fig/}}
\usetikzlibrary{arrows,automata,topaths,matrix,positioning,fit}
\usepgflibrary{shapes.geometric}
\usetikzlibrary{shapes.geometric}
\tikzset{every state/.style={minimum size=0pt}}
\usepackage{footnote}
\usepackage{float}
\usepackage{tabularx}
\usepackage{pifont}
\usepackage{txfonts}

\newtheorem{theorem}{Theorem}

\newtheorem{lemma}[theorem]{Lemma}
\newtheorem{proposition}[theorem]{Proposition}

\theoremstyle{definition} 

\newtheorem{convention}[theorem]{Convention}
\newtheorem{definition}[theorem]{Definition}

\newtheorem{example}[theorem]{Example}

\newtheorem{remark}[theorem]{Remark}

\newtheorem{warning}[theorem]{Warning}

\newrefformat{§}{§\ref{#1}}
\newrefformat{algorithm}{Algorithm \ref{#1}}
\newrefformat{a}{Answer \ref{#1}}
\newrefformat{appendix}{Appendix \ref{#1}}
\newrefformat{claim}{Claim \ref{#1}}
\newrefformat{conclusion}{Conclusion \ref{#1}}
\newrefformat{convention}{Convention \ref{#1}}
\newrefformat{conjecture}{Conjecture \ref{#1}}
\newrefformat{c}{Corollary \ref{#1}}
\newrefformat{equ}{(\ref{#1})}
\newrefformat{e}{Example \ref{#1}}
\newrefformat{exe}{Exercise \ref{#1}}
\newrefformat{d}{Definition \ref{#1}}
\newrefformat{done}{Done \ref{#1}}
\newrefformat{f}{Fact \ref{#1}}
\newrefformat{fig}{Figure \ref{#1}}
\newrefformat{h}{Hypothesis \ref{#1}}
\newrefformat{idea}{Idea \ref{#1}}
\newrefformat{i}{Item \ref{#1}}
\newrefformat{l}{Lemma \ref{#1}}
\newrefformat{note}{Note \ref{#1}}
\newrefformat{n}{Notation \ref{#1}}
\newrefformat{o}{Observation \ref{#1}}
\newrefformat{problem}{Problem \ref{#1}}
\newrefformat{p}{Proposition \ref{#1}}
\newrefformat{pseudocode}{Pseudocode \ref{#1}}
\newrefformat{Q}{Q. \ref{#1}}
\newrefformat{r}{Remark \ref{#1}}
\newrefformat{table}{Table \ref{#1}}
\newrefformat{t}{Theorem \ref{#1}}
\newrefformat{Todo}{Todo \ref{#1}}
\newrefformat{w}{Warning \ref{#1}}

\newcommand{\righttherefore}{:\joinrel\cdot\,}

\title{
	Similarity-based analogical proportions
}
\author{
	Christian Anti\'c
}
\address{
	christian.antic@icloud.com\\
	Vienna University of Technology\\
	Vienna, Austria
}

\begin{document}

\begin{abstract}
	The author has recently introduced abstract algebraic frameworks of analogical proportions and similarity within the general setting of universal algebra. The purpose of this paper is to build a bridge from similarity to analogical proportions by formulating the latter in terms of the former. The benefit of this similarity-based approach is that the connection between proportions and similarity is built into the framework and therefore evident which is appealing since proportions and similarity are both at the center of analogy; moreover, future results on similarity can directly be applied to analogical proportions.
\end{abstract}

\maketitle

\section{Introduction and preliminaries}

The author has recently introduced an abstract algebraic framework of analogical proportions of the form ``$a$ is to $b$ what $c$ is to $d$'' written $a:b::c:d$ in the general setting of universal algebra with appealing mathematical properties \citeA{Antic22}. It has recently been applied to logic program synthesis in \citeA{Antic23-23} and it has been studied for monounary algebras in \citeA{Antic22-2}.

In parallel, the author has recently introduced an abstract algebraic model of similarity based on the idea that the set of generalizations of an element contains important information about the properties of that element \cite{Antic23-2}. For example, the term $2x$ is a generalization of an integer $a$ iff $a$ is even, and $x^2$ is a generalization of $a$ iff $a$ is a square number.

The \textbf{purpose of this paper} is to combine the two aforementioned frameworks by defining analogical proportions in terms of similarity thus putting similarity at the center of proportions. The benefit of this similarity-based approach is that the connection between proportions and similarity is built into the framework and therefore evident --- this is not the case for the previous definition given in \citeA{Antic22}. This is appealing since proportions and similarity are \textit{both} at the center of analogy. Most importantly, it allows us to directly apply future results obtained for similarity to analogical proportions.

We compare the two mentioned approaches in \prettyref{§:Comparison} and notice that there are subtle differences. Concretely, while the framework in \citeA{Antic22} always satisfies inner p-reflexivity $a:a::c:c$, it fails in the similarity-based framework of this paper justified by a reasonable counterexample in \prettyref{e:aabb}. Moreover, the Uniqueness Lemma, which is a key result in \citeA{Antic22}, fails here as well (see \prettyref{w:UL}). On the other hand, in \prettyref{§:ITs} we show that the Isomorphism Theorems in \citeA{Antic22} can be smoothly transferred to similarity-based setting showing that proportions are compatible with structure-preserving mappings.


We assume the reader to be fluent in basic universal algebra as it is presented for example in \citeA[§II]{Burris00}.

A \textit{\textbf{language}} $L$ of algebras is a set of \textit{\textbf{function symbols}}\footnote{We omit constant symbols as we identify constants with 0-ary functions.} together with a \textit{\textbf{rank function}} $r:L\to\mathbb N$, and a denumerable set $X$ of \textit{\textbf{variables}} distinct from $L$. Terms are formed as usual from variables in $X$ and function symbols in $L$, and we denote the set of all such $LX$-terms by $T_{L,X}$. We denote the variables occurring in a term $s$ by $Xs$. The \textit{\textbf{rank}} of a term $s$ is given by the number of its variables and denoted by $rs$.

An \textit{\textbf{$L$-algebra}} $\mathfrak A$ consists of a non-empty set $A$, the \textit{\textbf{universe}} of $\mathfrak A$, and for each function symbol $f\in L$, a function $f^\mathfrak A:A^{rf}\to A$, the \textit{\textbf{functions}} of $\mathfrak A$ (the \textit{\textbf{distinguished elements}} of $\mathfrak A$ are the 0-ary functions). Every term $s$ induces a function $s^\mathfrak A$ on $\mathfrak A$ in the usual way.  We call a term $t$ \textit{\textbf{injective}} in $\mathfrak A$ iff $t^\mathfrak A$ is an injective function.

A \textit{\textbf{homomorphism}} is a mapping $H: \mathfrak{A\to B}$ such that for any function symbol $f\in L$ and elements $a_1,\ldots,a_{rf}\in A$,
\begin{align*} 
	Hf^ \mathfrak A(a_1,\ldots,a_{rf})=f^ \mathfrak B(Ha_1,\ldots,Ha_{rf}).
\end{align*} An \textit{\textbf{isomorphism}} is a bijective homomorphism.

In the rest of the paper, $L$ denotes a language of algebras and $\mathfrak A$ and $\mathfrak B$ are $L$-algebras.

\section{Similarity-based analogical proportions}\label{§:SAPs}

This is the main section of the paper. Here we shall introduce similarity-based analogical proportions based on algebraic similarity, which we shall now briefly recall.

\begin{definition} A \textit{\textbf{generalization}} of an element $a\in A$ in $\mathfrak A$ is an $LX$-term $s$ such that $a=s^ \mathfrak A\textbf{o}$, for some $\mathbf o\in A^{rs}$, and we denote the set of all such generalizations by $\uparrow_ \mathfrak A a$. Moreover, we define, for elements $a\in A$ and $b\in B$,
\begin{align*} 
	a\uparrow_{ \mathfrak{(A,B)}} b:=(\uparrow_ \mathfrak A a)\cap (\uparrow_ \mathfrak B b).
\end{align*} We will often omit the indices in case the underlying algebras are understood.
\end{definition}

\begin{definition}[\citeA{Antic23-2}] We define the \textit{\textbf{similarity relation}} as follows:
\begin{enumerate}
	\item A generalization is \textit{\textbf{trivial}} in $\mathfrak{(A,B)}$ iff it generalizes all elements in $A$ and $B$ and we denote the set of all such trivial generalizations by $\emptyset_ \mathfrak{(A,B)}$. 

	\item Now we say that $a\lesssim b$ \textit{\textbf{holds}} in $\mathfrak{(A,B)}$ --- in symbols,
	\begin{align*} 
		a\lesssim_{ \mathfrak{(A,B)}} b,
	\end{align*} iff
	\begin{enumerate}
		\item either $(\uparrow_ \mathfrak A a)\cup (\uparrow_ \mathfrak B b)$ consists only of trivial generalizations; or
		\item $a\uparrow_{ \mathfrak{(A,B)}} b$ contains at least one non-trivial generalization and is maximal with respect to subset inclusion among the sets $a\uparrow_{ \mathfrak{(A,B)}} b'$, $b'\neq a\in B$, that is, for any element $b'\neq a\in B$,
		\begin{align*} 
			\emptyset_ \mathfrak{(A,B)}\subsetneq a\uparrow_{ \mathfrak{(A,B)}} b\subseteq a\uparrow_{ \mathfrak{(A,B)}} b'
		\end{align*} implies
		\begin{align*} 
			\emptyset_ \mathfrak{(A,B)}\subsetneq a\uparrow_{ \mathfrak{(A,B)}} b'\subseteq a\uparrow_{ \mathfrak{(A,B)}} b.
		\end{align*} We abbreviate the above definition by simply saying that $a\uparrow_{ \mathfrak{(A,B)}} b$ is \textit{\textbf{$b$-maximal}}.
	\end{enumerate}

	\item Finally, the \textit{\textbf{similarity relation}} is defined as
	\begin{align*} 
		a\approx_\mathfrak{(A,B)} b 
			\quad:\Leftrightarrow\quad a\lesssim_{ \mathfrak{(A,B)}} b \quad\text{and}\quad b\lesssim_{\mathfrak{(B,A)}} a,
	\end{align*} in which case we say that $a$ and $b$ are \textit{\textbf{similar}} in $\mathfrak{(A,B)}$.
\end{enumerate} We will always write $\mathfrak A$ instead of $\mathfrak{(A,A)}$.
\end{definition}

We now analyze the following basic properties:
\begin{align*} 
	& a\approx_ \mathfrak A a \quad\text{(reflexivity)},\\
	& a\approx_ \mathfrak{(A,B)} b \quad\Leftrightarrow\quad b\approx_{\mathfrak{(B,A)}} a \quad\text{(symmetry)},\\
	& a\approx_ \mathfrak{(A,B)} b \quad\text{and}\quad b\approx_\mathfrak{(B,C)} c \quad\Rightarrow\quad a\approx_\mathfrak{(A,C)} c \quad\text{(transitivity)}. 
\end{align*}

\begin{proposition}[\citeA{Antic23-2}]\label{p:g-properties} The similarity relation is reflexive, symmetric, and in general not transitive.
\end{proposition}

Computing all generalizations is often difficult which fortunately can be omitted in many cases:

\begin{definition}[\citeA{Antic23-2}]\label{d:char_jus_lesssim} We call a set $G$ of $LX$-terms a \textit{\textbf{characteristic set of generalizations}} of $a \lesssim b$ in $\mathfrak{(A,B)}$ iff $G$ is a sufficient set of generalizations, that is, iff
\begin{enumerate}
	\item $G\subseteq a\uparrow_\mathfrak{(A,B)} b$, and
	\item $G\subseteq a\uparrow_\mathfrak{(A,B)}b'$ implies $b'=b$, for each $b'\neq a\in B$.
\end{enumerate} In case $G=\{s\}$ is a singleton set satisfying both conditions, we call $s$ a \textit{\textbf{characteristic generalization}} of $a\lesssim b$ in $\mathfrak{(A,B)}$.
\end{definition}

We now wish to define analogical proportions in terms of algebraic similarity. For this, we need the following auxiliary notion:

\begin{definition} We define the $L$-algebra $Arr( \mathfrak A)$ as follows: the universe of $Arr( \mathfrak A)$ consists of \textit{\textbf{arrows}} $a\to b$ with $a,b\in A$, and its functions are given by the functions of $\mathfrak A$ generalized to arrows component-wise. 
\end{definition}

\begin{definition} The \textit{\textbf{set of generalizations}} of an arrow $a\to b$ in $Arr( \mathfrak A)$ is given by
\begin{align*} 
	\uparrow_{Arr( \mathfrak A)}(a\to b)= \left\{s\to t\in T_{L,X}\to T_{L,X} \;\middle|\; a\to b=s^ \mathfrak A\textbf{o}\to t^ \mathfrak A\textbf{o},\text{ for some $\mathbf o\in A^{rs}$} \right\}.
\end{align*} We call elements of $(a\to b)\uparrow_{ \mathfrak{(A,B)}} (c\to d)$ \textit{\textbf{justifications}} of $a\to b\lesssim c\to d$ in $\mathfrak{(A,B)}$.
\end{definition}

\begin{convention} We make the convention that $\to$ binds weaker than every other algebraic operation except for $\lesssim$ and $\approx$ which have the weakest binding.
\end{convention}

We are now ready to introduce the main notion of the paper:

\begin{definition}\label{d:abcd} 

Given a pair of $L$-algebras $ \mathfrak{(A,B)}$, and elements $a,b\in A$ and $c,d\in B$, we define the \textit{\textbf{similarity-based analogical proportion relation}} in $\mathfrak{(A,B)}$ by\footnote{We read $a:b\approx c:d$ as ``$a$ is to $b$ is similar as $c$ is to $d$''.}
\begin{align*} 
	a:b\approx_\mathfrak{(A,B)} c:d \quad:\Leftrightarrow\quad 
		&a\to b\approx_{Arr( \mathfrak A)Arr( \mathfrak B)} c\to d,\\
		&b\to a\approx_{Arr( \mathfrak A)Arr( \mathfrak B)} d\to c,
\end{align*} which translates into
\begin{align*} 
	a:b\approx_\mathfrak{(A,B)} c:d \quad\Leftrightarrow\quad 
		&a\to b \lesssim_{Arr( \mathfrak A)Arr( \mathfrak B)} c\to d,\\
		&b\to a \lesssim_{Arr( \mathfrak A)Arr( \mathfrak B)} d\to c,\\
		&c\to d \lesssim_{Arr( \mathfrak B)Arr( \mathfrak A)} a\to b,\\ 
		&d\to c \lesssim_{Arr( \mathfrak B)Arr( \mathfrak A)} b\to a.
\end{align*}
\end{definition}

\begin{convention} With a slight abuse of notation, in what follows we will not distinguish between $\mathfrak A$ and $Arr( \mathfrak A)$ and we will often omit the reference to the underlying algebra in case it is understood from the context.
\end{convention}

Notice that \prettyref{d:char_jus_lesssim} directly yields a notion of a characteristic set of justifications (see \prettyref{d:char_jus}):

\begin{definition} A set $J$ of justifications of $a\to b\lesssim c\to d$ is a \textit{\textbf{characteristic set of justifications}} iff
\begin{itemize}
	\item $J\subseteq (a\to b)\uparrow_{ \mathfrak{(A,B)}}(c\to d)$;
	\item $J\subseteq (a\to b)\uparrow_{ \mathfrak{(A,B)}}(c'\to d')$ implies $c'\to d'=c\to d$, for all arrows $c'\to d'$ in $\mathfrak B$.
\end{itemize} In case $J=\{s\to t\}$ is a singleton set satisfying both conditions, we call $s\to t$ a \textit{\textbf{characteristic justification}} of $a\to b\lesssim c\to d$ in $\mathfrak{(A,B)}$.
\end{definition}

\begin{example}\label{e:exa} First consider the algebra $\mathfrak A_1:=(\{a,b,c,d\})$, consisting of four distinct elements with no functions and no constants:
\begin{center}
\begin{tikzpicture} 
    \node (a)               {$a$};
    \node (b) [above=of a]  {$b$};
    \node (c) [right=of a]  {$c$};
    \node (d) [above=of c]  {$d$};
\end{tikzpicture}
\end{center} Since $\uparrow_{\mathfrak A_1}(x\to y)\ \cup \uparrow_{\mathfrak A_1}(x'\to y')$ contains only trivial justifications for \textit{any distinct} elements $x,x',y,y'\in A'$, we have, for example:
\begin{align*} 
    a:b\approx_{\mathfrak A_1} c:d \quad\text{and}\quad a:c\approx_{\mathfrak A_1} b:d.
\end{align*} On the other hand, since
\begin{align*} 
    \uparrow_{\mathfrak A_1}(a\to a)\ \cup \uparrow_{\mathfrak A_1}(a\to d)=\{x\to x\}\neq\emptyset
\end{align*} and
\begin{align*} 
    \emptyset=\ \uparrow_{\mathfrak A_1}(a\to a\righttherefore a\to d)\subsetneq\ \uparrow_{\mathfrak A_1}(a\to a\righttherefore a\to a)=\{x\to x\},
\end{align*} we have
\begin{align*} 
    a\to a\not\lesssim_{\mathfrak A_1} a\to d,
\end{align*} which implies
\begin{align*} 
    a:a\not\approx_{\mathfrak A_1} a:d.
\end{align*}

Now consider the slightly different algebra $\mathfrak A_2:=(\{a,b,c,d\},f)$, where $f$ is the unary function defined by (we omit the loops $fx:=x$ for $x\in\{b,c,d\}$ in the figure):
\begin{center}
\begin{tikzpicture} 
	\node (a)               {$a$};
	\node (b) [above=of a,yshift=1cm]  {$b$};
	\node (c) [right=of a,xshift=1cm]  {$c$};
	\node (d) [right=of b,xshift=1cm]  {$d$};

	\draw[->] (a) to [edge label'={$f$}] (b);
\end{tikzpicture}
\end{center} We expect $a:b\approx c:d$ to fail in $\mathfrak A_2$ as it has no non-trivial justification. In fact, 
\begin{align*} \uparrow_{\mathfrak A_2}(a\to b)\ \cup \uparrow_{\mathfrak A_2}(c\to d)=\left\{x\to f^n x \;\middle|\; n\geq 1\right\}\neq\emptyset 
\end{align*} whereas
\begin{align*} 
    (a\to b)\uparrow_{\mathfrak A_2}(c\to d)=\emptyset
\end{align*} show
\begin{align*} a:b\not\approx_{\mathfrak A_2} c:d.
\end{align*}

In the algebra $\mathfrak A_3$ given by
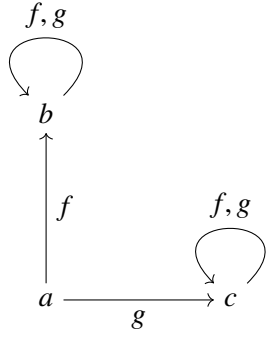
\begin{center}
\begin{tikzpicture} 
    \node (a)               {$a$};
    \node (b) [above=of a,yshift=1cm]  {$b$};
    \node (c) [right=of a,xshift=1cm]  {$c$};
    \draw[->] (a) to [edge label'={$f$}] (b);
    \draw[->] (a) to [edge label'={$g$}] (c);
    \draw[->] (b) to [edge label'={$f,g$}] [loop] (b);
    \draw[->] (c) to [edge label'={$f,g$}] [loop] (c);
\end{tikzpicture}
\end{center} we have
\begin{align*} 
    a:b\not\approx_{\mathfrak A_3} a:c.
\end{align*} The intuitive reason is that $a:b\approx b:b$ is a more plausible proportion than $a:b\approx a:c$, which is reflected in the following computation:
\begin{align*} 
    \emptyset=(a\to b)\uparrow_{\mathfrak A_3}(a\to c)\subsetneq (a\to b)\uparrow_{\mathfrak A_3}(b\to b)=\{x\to fx,\ldots\}.
\end{align*}
\end{example}

\section{Properties}\label{§:Properties}

In the tradition of the ancient Greeks, \citeA{Lepage03} introduced (in the linguistic context) a set of properties as a guideline for formal models of analogical proportions and his list has since been modified and extended by a number of authors and can now be summarized as follows:\footnote{\citeA{Lepage03} uses different names for his properties --- we have decided to remain consistent with the nomenclature in \citeA[§4.2]{Antic22}.} 
\begin{align}
    &a:b::_ \mathfrak A a:b \quad\text{(p-reflexivity)},\\
    &a:b::_\mathfrak{(A,B)}c:d \quad\Leftrightarrow\quad c:d::_{\mathfrak{(B,A)}}a:b\quad\text{(p-symmetry)},\\
    &a:b::_\mathfrak{(A,B)}c:d \quad\Leftrightarrow\quad b:a::_\mathfrak{(A,B)}d:c\quad\text{(inner p-symmetry)},\\
    &a:a::_\mathfrak A a:d \quad\Leftrightarrow\quad d=a\quad\text{(p-determinism)},\\
    &a:a::_\mathfrak{(A,B)}c:c \quad\text{(inner p-reflexivity)},\\
    &a:b::_ \mathfrak A c:d \quad\Leftrightarrow\quad a:c::_ \mathfrak A b:d \quad\text{(central permutation)},\\
    &a:a::_ \mathfrak A c:d \quad\Rightarrow\quad d=c \quad\text{(strong inner p-reflexivity)},\\
    &a:b::_ \mathfrak A a:d \quad\Rightarrow\quad d=b \quad\text{(strong p-reflexivity)}.
\end{align} Moreover, the following property is considered, for $a,b\in A\cap B$:
\begin{align}
    a:b::_\mathfrak{(A,B)} b:a\quad\text{(p-commutativity).}
\end{align} 

Furthermore, the following properties are considered, for $L$-algebras $\mathfrak{A,B,C}$ and elements $a,b\in A$, $c,d\in B$, $e,f\in C$:
\begin{prooftree}
    \AxiomC{$a:b::_\mathfrak{(A,B)} c:d$}
        \AxiomC{$c:d::_\mathfrak{(B,C)} e:f$}
        \RightLabel{(p-transitivity),}
    \BinaryInfC{$a:b::_\mathfrak{(A,C)} e:f$}
\end{prooftree} and, for elements $a,b,e\in A$ and $c,d,f\in B$, the property
\begin{prooftree}
    \AxiomC{$a:b::_\mathfrak{(A,B)} c:d$}
    \AxiomC{$b:e::_\mathfrak{(A,B)} d:f$}
    \RightLabel{(inner p-transitivity),}
    \BinaryInfC{$a:e::_\mathfrak{(A,B)} c:f$}
\end{prooftree} and, for elements $a\in A$, $b\in A\cap B$, $c\in B\cap C$, and $d\in C$, the property
\begin{prooftree}
    \AxiomC{$a:b::_\mathfrak{(A,B)}b:c$}
    \AxiomC{$b:c::_\mathfrak{(B,C)}c:d$}
    \RightLabel{(central p-transitivity).}
    \BinaryInfC{$a:b::_\mathfrak{(A,C)}c:d$}
\end{prooftree} Notice that central p-transitivity follows from p-transitivity. 

We have the following analysis of the proportional axioms within the framework of this paper:

\begin{theorem}\label{t:properties} The similarity-based analogical proportion relation as defined in \prettyref{d:abcd} satisfies
\begin{itemize}
    \item p-reflexivity,
    \item p-symmetry,
    \item inner p-symmetry,
    \item p-determinism,
\end{itemize}  and, in general, it does not satisfy
\begin{itemize}
    \item inner p-reflexivity,
    \item central permutation,
    \item strong inner p-reflexivity,
    \item strong p-reflexivity,
    \item p-commutativity,
    \item p-transitivity,
    \item inner p-transitivity,
    \item central p-transitivity.
\end{itemize}
\end{theorem}
\begin{proof} We have the following positive proofs:
\begin{itemize}
    \item p-Symmetry and p-reflexivity follow from \prettyref{p:g-properties}.
    \item Inner p-reflexivity is am immediate consequence of the definition.
    \item Next, we prove p-determinism. ($\Leftarrow$) Inner p-reflexivity implies
    \begin{align*} 
        a:a\approx a:a.
    \end{align*} $(\Rightarrow)$ Since $x\to x\in\ \uparrow(a\to a)$, the set $$\uparrow(a\to a)\ \cup \uparrow(a\to d)$$ cannot consist only of trivial justifications. We clearly have
    \begin{align*} 
		(a\to a)\uparrow(a\to d)
			&=\ \uparrow(a\to a)\ \cap\uparrow(a\to d)\\
			&\subseteq\ \uparrow(a\to a)\\
			&=(a\to a)\uparrow(a\to a).
	\end{align*} On the other hand, we have
    \begin{align*} 
        x\to x\in (a\to a)\uparrow(a\to a)
    \end{align*} whereas
    \begin{align*} 
        x\to x\not\in (a\to a)\uparrow(a\to d),\quad\text{for all $d\neq a$}.
    \end{align*} This shows
    \begin{align*} 
        (a\to a)\uparrow(a\to d)\subsetneq (a\to a)\uparrow (a\to a),
    \end{align*} which implies
    \begin{align*} 
        a:a\not\approx a:d,\quad\text{for all $d\neq a$.}
    \end{align*}
\end{itemize} 

We have the following negative proofs (all but the first one are similar to the corresponding proofs of Theorem 28 in \citeA{Antic22}):
\begin{itemize}
    \item Inner p-reflexivity fails, for example, in the algebra $(\{a,b,c\},f)$ given by
    
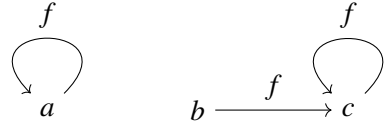
\begin{center}
	\begin{tikzpicture} 
	    \node (a) {$a$};
	    \node (b) [right of=a,xshift=1cm] {$b$};
	    \node (c) [right of=b,xshift=1cm] {$c$};
	    \draw[->] (a) to [edge label'={$f$}][loop] (a);
	    \draw[->] (b) to [edge label={$f$}] (c);
	    \draw[->] (c) to [edge label'={$f$}][loop] (c);
	\end{tikzpicture}
	\end{center} since we have
	\begin{align*} 
	    (a\to a)\uparrow(b\to b)=\{x\to x\}\subsetneq\left\{f^m x\to f^n x \;\middle|\; m,n\geq 0\right\}=(a\to a)\uparrow(c\to c)
	\end{align*} and thus $a\to a\not\lesssim b\to b$ which implies $a:a\not\approx b:b$.

    \item Central permutation fails, for example, in the algebra $(\{a,b,c,d\},f)$ given by (we omit the loops $fx:=x$ for $x\in\{b,c,d\}$ in the figure)
    
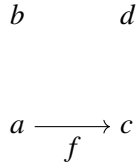
\begin{center}
        \begin{tikzpicture} 
        \node (a)               {$a$};
        \node (b) [above=of a]  {$b$};
        \node (c) [right=of a]  {$c$};
        \node (d) [right=of b]  {$d$};
        \draw[->] (a) to [edge label'={$f$}] (c);
    \end{tikzpicture}
    \end{center} since we clearly have
    \begin{align*} 
		a:b\approx c:d \quad\text{whereas}\quad	a:c\not\approx b:d.
	\end{align*}

    \item Strong inner p-reflexivity fails, for example, in the algebra $(\{a,c,d\},f)$ given by
    
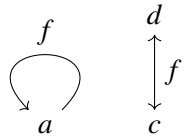
\begin{center}
    \begin{tikzpicture} 
        \node (a)               {$a$};
        \node (c) [right=of a]  {$c$};
        \node (d) [above=of c]  {$d$};
        \draw[->] (a) to [edge label'={$f$}] [loop] (a);
        \draw[<->] (c) to [edge label'={$f$}] (d);
    \end{tikzpicture}
    \end{center} since
    \begin{align*} 
		\uparrow(c\to d)=\ \uparrow(d\to c)
	\end{align*} shows that
	\begin{align*} 
		(a\to a)\uparrow (c\to d)&=\ \uparrow(c\to d),\\
		(a\to a)\uparrow (d\to c)&=\ \uparrow(d\to c),
	\end{align*} and thus
	\begin{align*} 
		a:a\approx c:d.
	\end{align*}

    \item Strong p-reflexivity fails, for example, in the algebra $(\{a,b,c\})$ having no functions, since we have
    \begin{align*} 
		a:b\approx_{(\{a,b,c\})} a:d.
	\end{align*}

    \item p-Commutativity fails in the algebra $(\{a,b\},f)$ given by
    
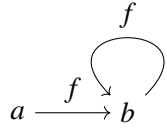
\begin{center}
    \begin{tikzpicture} 
        \node (a)               {$a$};
        \node (b) [right=of a]  {$b$};
        \draw[->] (a) to [edge label={$f$}] (b);
        \draw[->] (b) to [edge label'={$f$}] [loop] (b);
    \end{tikzpicture}
    \end{center} since
    \begin{align*}
        \uparrow(a\to b)\ \cup \uparrow(b\to a)=\{x\to fx,\ldots\}\neq\emptyset,
    \end{align*} whereas
    \begin{align*} 
        (a\to b)\uparrow(b\to a)=\emptyset.
    \end{align*}

    \item p-Transitivity fails in the algebra $(\{a,b,c,d,e,f\},g,h)$ given by (we omit the loops $gx:=x$ for $x\in\{b,d,e,f\}$, and $hx:=x$ for $x\in\{a,b,d,f\}$, in the figure)
    \begin{center}
    \begin{tikzpicture} 
        \node (a)               {$a$};
        \node (b) [right=of a]  {$b$};
        \node (c) [right=of b, xshift=0.8cm]  {$c$};
        \node (d) [right=of c]  {$d$};
        \node (e) [right=of d, xshift=0.8cm]  {$e$};
        \node (f) [right=of e]  {$f$.};
        \draw[->] (a) to [edge label={$g$}] (b);
        \draw[->] (c) to [edge label={$g,h$}] (d);
        \draw[->] (e) to [edge label={$h$}] (f);
    \end{tikzpicture}
    \end{center} The arrow proportions $a\to b\lesssim c\to d$ and $c\to d\lesssim a\to b$ follow from the maximality of
    \begin{align*} 
		(a\to b)\uparrow (c\to d)= \left\{x\to g^nx \;\middle|\; n\geq 1\right\},
	\end{align*} and the arrow proportions $b\to a\lesssim d\to c$ and $d\to c\lesssim b\to a$ follow from the maximality of
    \begin{align*} 
        (b\to a)\uparrow (d\to c)= \left\{g^nx\to x \;\middle|\; n\geq 0\right\}.
    \end{align*} This shows
    \begin{align*} 
		a:b\approx c:d.
	\end{align*} An analogous argument shows
	\begin{align*} 
		c:d\approx e:f.
	\end{align*} On the other hand,
    \begin{align*} 
        \uparrow(a\to b)\ \cup \uparrow(e\to f)\neq\emptyset \quad\text{whereas}\quad (a\to b)\uparrow(e\to f)=\emptyset
    \end{align*} shows
    \begin{align*} 
		a:b\not\approx e:f.
	\end{align*}

    \item Inner p-transitivity fails, for example, in the algebra $(\{a,b,c,d,e,f\},g)$ given by (we omit the loops $gx:=x$, for $x\in\{b,e,c,d,f\}$, in the figure)
    
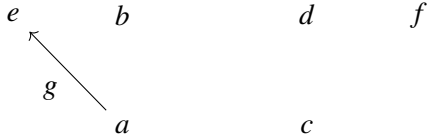
\begin{center}
    \begin{tikzpicture} 
        \node (a) {$a$};
        \node (b) [above=of a] {$b$};
        \node (e) [left=of b] {$e$};
        \node (c) [right=of a,xshift=1cm] {$c$};
        \node (d) [above=of c] {$d$};
        \node (f) [right=of d] {$f$};
        \draw[->] (a) to [edge label={$g$}] (e);
    \end{tikzpicture}
    \end{center} We clearly have
    \begin{align*} 
		a:b\approx c:d \quad\text{and}\quad b:e\approx d:f.
	\end{align*} On the other hand, 
    \begin{align*} 
        \uparrow(a\to e)\ \cup \uparrow(c\to f)\neq\emptyset \quad\text{whereas}\quad (a\to e)\uparrow (c\to f)=\emptyset
    \end{align*} shows
    \begin{align*} 
		a:e\not\approx c:f.
	\end{align*}

    \item Central p-transitivity fails, for example, in the algebra $(\{a,b,c,d\},g,h)$ given by (we omit the loops $gx:=x$ for $x\in\{c,d\}$, and $hx:=x$ for $x\in\{a,d\}$, in the figure)
    \begin{center}
    \begin{tikzpicture} 
        \node (a)               {$a$};
        \node (b) [right=of a]  {$b$};
        \node (c) [right=of b]  {$c$};
        \node (d) [right=of c]  {$d$};
        \draw[->] (a) to [edge label={$g$}] (b);
        \draw[->] (b) to [edge label={$g,h$}] (c);
        \draw[->] (c) to [edge label={$h$}] (d);
    \end{tikzpicture}
    \end{center} The proof is analogous to the above disproof of p-transitivity.
\end{itemize}
\end{proof}


\section{Isomorphism Theorems}\label{§:ITs}

In this section, we transfer the Isomorphism and Homomorphism Theorems in \citeA{Antic22,Antic23-22} to the similarity-based setting. For this, we first recall the following lemma; see \citeA[Homomorphism Lemma]{Antic23-22} and \citeA[Isomorphism Lemma]{Antic22}:

\begin{lemma}[Isomorphism Lemma]\label{l:IL} For any homomorphism $H: \mathfrak{A\to B}$ and $a,b\in A$,
\begin{align}\label{equ:ab_subseteq_HaHb}  
    \uparrow_ \mathfrak A(a\to b)\subseteq\ \uparrow_ \mathfrak B(Ha\to Hb).
\end{align} In case $H$ is an isomorphism, we have
\begin{align}\label{equ:ab=HaHb} 
    \uparrow_ \mathfrak A(a\to b)=\ \uparrow_ \mathfrak B(Ha\to Hb).
\end{align}
\end{lemma}


\begin{theorem}[First Isomorphism Theorem]\label{t:FIT} For any homomorphism $H: \mathfrak{A\to B}$ and elements $a,b\in A$, we have the following implication:
\begin{prooftree}
    \AxiomC{$\uparrow_ \mathfrak A(a\to b)=\emptyset \quad\Rightarrow\quad \uparrow_ \mathfrak B(Ha\to Hb)=\emptyset$}
    \RightLabel{.}
    \UnaryInfC{$a\to b \lesssim_\mathfrak{(A,B)} Ha\to Hb$}
\end{prooftree} In case $H$ is an isomorphism, we have
\begin{align}\label{equ:abHaHb} 
    a:b\approx_\mathfrak{(A,B)} Ha:Hb,\quad\text{for all $a,b\in A$}.
\end{align}
\end{theorem}
\begin{proof} The first implication is shown in essentially the same way as the same implication in the proof of \citeA[Homomorphism Theorem]{Antic23-22}: By the Isomorphism \prettyref{l:IL} we have
\begin{align*} 
    \uparrow_{ \mathfrak{(A,B)}}(a\to b\righttherefore Ha\to Hb)=\ \uparrow_ \mathfrak A(a\to b)\ \cap \uparrow_ \mathfrak B(Ha\to Hb)=\ \uparrow_ \mathfrak A(a\to b),
\end{align*} which shows the $(Ha\to Hb)$-maximality of $\uparrow_{ \mathfrak{(A,B)}}(a\to b\lesssim Ha\to Hb)$. It remains to show that we cannot have
\begin{align*} 
    \uparrow_ \mathfrak A(a\to b)\ \cup\uparrow_ \mathfrak B(Ha\to Hb)\neq\emptyset \quad\text{whereas}\quad \uparrow_{ \mathfrak{(A,B)}}(a\to b\lesssim Ha\to Hb)=\emptyset.
\end{align*} This is a direct consequence of \prettyref{equ:ab_subseteq_HaHb} and the assumption that $\uparrow_ \mathfrak A(a\to b)=\emptyset$ implies $\uparrow_ \mathfrak B(Ha\to Hb)=\emptyset$.

Now assume $H$ is an isomorphism. The proof of \prettyref{equ:abHaHb} is essentially the same as the proof of the First Isomorphism Theorem in \citeA{Antic22}: 

If $\uparrow_\mathfrak A(a\to b)\ \cup \uparrow_\mathfrak B(Ha\to Hb)$ consists only of trivial justifications, we have $a\to b\lesssim Ha\to Hb$. 

Otherwise, there is at least one non-trivial justification $s\to t$ in $\uparrow_\mathfrak A(a\to b)$ or in $\uparrow_\mathfrak B(Ha\to Hb)$, in which case the Isomorphism \prettyref{l:IL} implies that $s\to t$ is in both $\uparrow_\mathfrak A(a\to b)$ \textit{and} $\uparrow_\mathfrak B(Ha\to Hb)$, which means that $(a\to b)\uparrow_\mathfrak{(A,B)}(Ha\to Hb)$ contains at least one non-trivial justification as well. 

We proceed by showing that $(a\to b)\uparrow_\mathfrak{(A,B)}(Ha\to Hb)$ is $(Ha\to Hb)$-maximal:
\begin{align*} 
    (a\to b)\uparrow_\mathfrak{(A,B)}(Ha\to Hb)
    &=\ \uparrow_\mathfrak A(a\to b)\ \cap \uparrow_\mathfrak B(Ha\to Hb)\\
    &\stackrel{ \prettyref{equ:ab=HaHb}}=\ \uparrow_\mathfrak A(a\to b)\\
    &\supseteq\ \uparrow_\mathfrak A(a\to b)\ \cap \uparrow_\mathfrak B(c\to d)\\
    &=(a\to b)\uparrow_\mathfrak{(A,B)}(c\to d),\quad\text{for every $c,d\in B$.}
\end{align*} This shows
\begin{align*} 
	a\to b\lesssim_\mathfrak{(A,B)} Ha\to Hb.
\end{align*} 

An analogous argument shows the remaining arrow proportions and thus \prettyref{equ:abHaHb}.
\end{proof}

\begin{remark} The Isomorphism \prettyref{t:FIT} shows that isomorphisms are \textit{\textbf{proportional analogies}} in the sense of \citeA{Antic22-4}.
\end{remark}

\begin{theorem}[Second Isomorphism Theorem]\label{t:SIT} For any isomorphism $H: \mathfrak{A\to B}$ and elements $a,b,c,d\in A$,
\begin{align*} 
    a:b\approx_ \mathfrak A c:d \quad\Leftrightarrow\quad Ha:Hb\approx_ \mathfrak B Hc:Hd,\quad\text{for all $a,b,c,d\in A$}.
\end{align*}
\end{theorem}
\begin{proof} A direct consequence of the Isomorphism \prettyref{l:IL}.
\end{proof}

\begin{remark} The Second Isomorphism \prettyref{t:SIT} shows that isomorphisms are \textit{\textbf{proportional isomorphisms}} in the sense of \citeA{Antic22-4}, and \textit{\textbf{analogy-preserving functions}} in the sense of \citeA[Definition 6]{Couceiro23}. Notice that it is slightly different from the Second Isomorphism Theorem in \citeA{Antic22}.
\end{remark}

\section{Comparison}\label{§:Comparison}

In this section, we first recall the abstract algebraic framework of analogical proportions in \citeA{Antic22} and then compare it to the framework of this paper.

\begin{convention} We will always write $s\twoheadrightarrow t$ instead of $(s,t)$, for any pair of $LX$-terms $s$ and $t$ such that every variable in $t$ occurs in $s$, that is, $Xt\subseteq Xs$. We call such expressions \textit{\textbf{$LX$-rewrite rules}} where we often omit the reference to $L$. We denote the set of all $LX$-rewrite rules with variables among $X$ by $R_{L,X}$. 
\end{convention}

\begin{definition}[\citeA{Antic22}]\label{d:abcd-22} We define the \textit{\textbf{analogical proportion entailment relation}} as follows:
\begin{enumerate}
    \item Define the \textit{\textbf{set of rewrite justifications}} (or \textit{\textbf{r-justifications}}) of an \textit{\textbf{arrow}} $a\to b$ in $\mathfrak A$ by
    \begin{align*} 
        Jus_\mathfrak A(a\to b):=\left\{s \twoheadrightarrow t\in R_{L,X} \;\middle|\; a\to b=s^\mathfrak A\textbf{o}\to t^\mathfrak A\textbf{o},\text{ for some }\mathbf o\in A^{|\mathbf x|}\right\},
    \end{align*} extended to an \textit{\textbf{arrow proportion}} $a\to b\righttherefore c\to d$ --- read as ``$a$ transforms into $b$ as $c$ transforms into $d$'' --- in $\mathfrak{(A,B)}$ by
    \begin{align*} 
        Jus_\mathfrak{(A,B)}(a\to b\righttherefore c\to d):=Jus_\mathfrak A(a\to b)\cap Jus_\mathfrak B(c\to d).
    \end{align*} An r-justification is \textit{\textbf{trivial}} in $\mathfrak{(A,B)}$ iff it justifies every arrow proportion in $\mathfrak{(A,B)}$ and we again denote this set by $\emptyset_\mathfrak{(A,B)}$. 

    \item Now we say that $a\to b\righttherefore c\to d$ \textit{\textbf{holds}} in $\mathfrak{(A,B)}$ --- in symbols,
    \begin{align*} 
        a\to b\righttherefore_\mathfrak{(A,B)}\, c\to d
    \end{align*} iff
    \begin{enumerate}
        \item either $Jus_\mathfrak A(a\to b)\cup Jus_\mathfrak B(c\to d)=\emptyset_\mathfrak{(A,B)}$ consists only of trivial r-justifications, in which case there is neither a non-trivial relation from $a$ to $b$ in $\mathfrak A$ nor from $c$ to $d$ in $\mathfrak B$; or

        \item $Jus_\mathfrak{(A,B)}(a\to b\righttherefore c\to d)$ is maximal with respect to subset inclusion among the sets $Jus_\mathfrak{(A,B)}(a\to b\righttherefore c\to d')$, $d'\in B$, containing at least one non-trivial r-justification, that is, for any element $d'\in B$,
        \begin{align*} 
            \emptyset_\mathfrak{(A,B)}\subsetneq Jus_\mathfrak{(A,B)}(a\to b\righttherefore c\to d)&\subseteq Jus_\mathfrak{(A,B)}(a\to b\righttherefore c\to d')
        \end{align*} implies
        \begin{align*} 
            \emptyset_\mathfrak{(A,B)}\subsetneq Jus_\mathfrak{(A,B)}(a\to b\righttherefore c\to d')\subseteq Jus_\mathfrak{(A,B)}(a\to b\righttherefore c\to d).
        \end{align*} We abbreviate the above definition by simply saying that $Jus_\mathfrak{(A,B)}(a\to b\righttherefore c\to d)$ is \textit{\textbf{$d$-maximal}}.
    \end{enumerate}

    \item Finally, the analogical proportion entailment relation is most succinctly defined by
    \begin{align*} 
        a:b::_{ \mathfrak{(A,B)}}c:d \quad:\Leftrightarrow\quad 
            &a\to b\righttherefore_{ \mathfrak{(A,B)}}\, c\to d \quad\text{and}\quad b\to a\righttherefore_{ \mathfrak{(A,B)}}\, d\to c\\
            &c\to d\righttherefore_{ \mathfrak{(B,A)}}\, a\to b \quad\text{and}\quad d\to c\righttherefore_{ \mathfrak{(B,A)}}\, b\to a.
    \end{align*}
\end{enumerate}
\end{definition}

\begin{warning} A difference between similarity-based analogical proportions as defined in \prettyref{d:abcd} and analogical proportions as defined in \prettyref{d:abcd-22} is that the former operates with justifications of the form $s\to t$ where $s$ and $t$ are \textit{arbitrary} terms, whereas the latter operates with r-justifications of the form $s \twoheadrightarrow t$ where we require $Xt\subseteq Xs$.
\end{warning}

\begin{definition}\label{d:char_jus} We call a set $J$ of r-justifications a \textit{\textbf{characteristic set of r-justifications}} of $a\to b\righttherefore c\to d$ in $\mathfrak{(A,B)}$ iff $J$ is a sufficient set of r-justifications in the sense that
\begin{enumerate}
    \item $J\subseteq Jus_\mathfrak{(A,B)}(a\to b\righttherefore c\to d)$, and
    \item $J\subseteq Jus_\mathfrak{(A,B)}(a\to b\righttherefore c\to d')$ implies $d'=d$, for each $d'\in B$.
\end{enumerate} In case $J=\{s \twoheadrightarrow t\}$ is a singleton set satisfying both conditions, we call $s \twoheadrightarrow t$ a \textit{\textbf{characteristic r-justification}} of $a\to b\righttherefore c\to d$ in $\mathfrak{(A,B)}$.
\end{definition}

Define, for a term $s\in T_{L,X}$ and element $a\in A$, the set
\begin{align*} 
    \langle s,a\rangle_\mathfrak A:=\left\{\mathbf o\in A^{rs} \;\middle|\; a=s^\mathfrak A\textbf{o}\right\},
\end{align*} consisting of all solutions to the polynomial equation $a=s\mathbf x$ in $\mathfrak A$. We can now depict every r-justification $s \twoheadrightarrow t$ of $a\to b\righttherefore c\to d$ as follows:
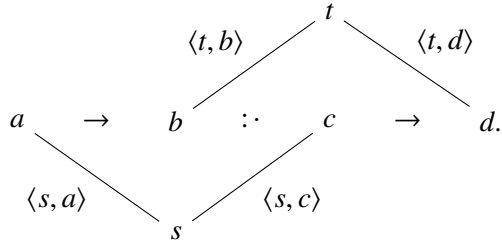
\begin{center}
\begin{tikzpicture}[node distance=1cm and 0.5cm]
    \node (a)               {$a$};
    \node (d1) [right=of a] {$\to$};
    \node (b) [right=of d1] {$b$};
    \node (d2) [right=of b] {$\righttherefore $};
    \node (c) [right=of d2] {$c$};
    \node (d3) [right=of c] {$\to$};
    \node (d) [right=of d3] {$d$.};
    \node (s) [below=of b] {$s$};
    \node (t) [above=of c] {$t$};

    \draw (a) to [edge label'={$\langle s,a\rangle$}] (s); 
    \draw (c) to [edge label={$\langle s,c\rangle$}] (s);
    \draw (b) to [edge label={$\langle t,b\rangle$}] (t);
    \draw (d) to [edge label'={$\langle t,d\rangle$}] (t);
\end{tikzpicture}
\end{center} Moreover, we have
\begin{align}\label{equ: 230824-langle} 
    s \twoheadrightarrow t\in Jus( a\to b\righttherefore c\to d) \quad\Leftrightarrow\quad \langle s,a\rangle\cap\langle t,b\rangle\neq\emptyset \quad\text{und}\quad \langle s,c\rangle\cap\langle t,d\rangle\neq\emptyset.
\end{align} 

The following examples show that there are subtle differences to be expected between the two notions of analogical proportions in Definitions \ref{d:abcd} and \ref{d:abcd-22}:

\begin{example}\label{e:aabb} Consider the algebra $(\{a,b,c\},f)$ given by
\begin{center}
\begin{tikzpicture} 
    \node (a) {$a$};
    \node (b) [right of=a,xshift=1cm] {$b$};
    \node (c) [right of=b,xshift=1cm] {$c$};
    \draw[->] (a) to [edge label'={$f$}][loop] (a);
    \draw[->] (b) to [edge label={$f$}] (c);
    \draw[->] (c) to [edge label'={$f$}][loop] (c);
\end{tikzpicture}
\end{center} Intuitively, it is reasonable to say that ``$a$ is to $a$ what $c$ is to $c$'', but it appears debatable whether the same should hold for $b$ instead of $c$ since the relation of $b$ to itself looks different. In fact, we shall now prove
\begin{align}\label{equ:aabb} 
    a:a::b:b \quad\text{whereas}\quad a:a\not\approx b:b.
\end{align} The first proportion is a direct consequence of inner p-reflexivity of non-similarity-based analogical proportions (\prettyref{t:properties2022}) and it basically follows from the fact that $x \twoheadrightarrow x$ is its characteristic r-justification. Why is $x\to x$ \textit{not} a characteristic justification of $a:a\approx b:b$? The reason is that we now don't look only at the last $b$ in $a\to a\righttherefore b\to b$, but at the arrow $b\to b$ as a whole and $x\to x$ justifies \textit{every} arrow $d\to d$ for any $d\in B$ --- hence, it is in general not a characteristic justification! Now since we have
\begin{align*} 
    (a\to a)\uparrow(b\to b)=\left\{f^mx\to x \;\middle|\; m\geq 0\right\}\subsetneq\left\{f^m x\to f^n x \;\middle|\; m,n\geq 0\right\}=(a\to a)\uparrow(c\to c),
\end{align*} we infer the second relation in \prettyref{equ:aabb}.
\end{example}

Define
\begin{align*} 
	\mathbbm 1_\mathfrak A(s):=\{a\in A\mid |\langle s,a\rangle_\mathfrak A|=1\}.
\end{align*}

A key result in the framework of \prettyref{d:abcd-22} is the following (we use here the formulation in \citeA{Antic23-22}):

\begin{lemma}[Uniqueness Lemma, \citeA{Antic22}]\label{l:UL} We have the following implications:
\begin{prooftree}
	\AxiomC{$s \twoheadrightarrow t\in Jus_\mathfrak{(A,B)}(a\to b\righttherefore c\to d)$}
		\AxiomC{$c\in\mathbbm 1_\mathfrak B(s)$}
	\BinaryInfC{$a\to b\righttherefore_\mathfrak{(A,B)}\, c\to d$}
\end{prooftree} and
\begin{prooftree}
	\AxiomC{$s \twoheadrightarrow t\in Jus_\mathfrak{(A,B)}(a\to b\righttherefore c\to d)$}
		\AxiomC{$a\in\mathbbm 1_\mathfrak A(s)\qquad b\in \mathbbm 1_\mathfrak A(t)\qquad c\in \mathbbm 1_\mathfrak B(s)\qquad d\in \mathbbm 1_\mathfrak B(t)$}
        \RightLabel{.}
	\BinaryInfC{$a:b::_\mathfrak{(A,B)}c:d$}
\end{prooftree}
\end{lemma}

\begin{warning}\label{w:UL} Notice that \prettyref{e:aabb} shows that we cannot adapt the Uniqueness \prettyref{l:UL} to the similarity-based setting since the term $x$ in the characteristic justification $x \twoheadrightarrow x$ of $a\to a \righttherefore b\to b$ is injective, and $x\to x$ is a justification of $a\to a\lesssim b\to b$, whereas $a\to a\not\lesssim b\to b$ shows that $x\to x$ is \textit{not} a characteristic justification in the similarity-based setting.
\end{warning}




Let us now turn our attention to the properties of \prettyref{§:Properties}:

\begin{theorem}[\citeA{Antic22}, Theorem 28]\label{t:properties2022} The analogical proportion relation as defined in \prettyref{d:abcd-22} satisfies
\begin{itemize}
    \item p-symmetry,
    \item inner p-symmetry,
    \item inner p-reflexivity,
    \item p-reflexivity,
    \item p-determinism,
\end{itemize}  and, in general, it does not satisfy
\begin{itemize}
    \item central permutation,
    \item strong inner p-reflexivity,
    \item strong p-reflexivity,
    \item p-commutativity,
    \item p-transitivity,
    \item inner p-transitivity,
    \item central p-transitivity,
    \item p-monotonicity.
\end{itemize}
\end{theorem}

\begin{warning} We see that the original framework of \citeA{Antic22} satisfies inner p-reflexivity, whereas the similarity-based framework does not (see \prettyref{t:properties}). Since the counterexample given in the proof of \prettyref{t:properties} is plausible, we interpret this discrepancy as a feature of the similarity-based approach of this paper.
\end{warning}



\section{Conclusion}

The purpose of this paper was to define analogical proportions in terms of the qualitative notion of algebraic similarity \cite{Antic23-2} within the general setting of universal algebra thus joining two concepts which are both at the center of analogy. We showed that most results in \citeA{Antic22} can be easily transferred. However, we have also seen that inner p-reflexivity $a:a\approx c:c$ fails in general (\prettyref{t:properties}) as justified by a reasonable counterexample in \prettyref{e:aabb} and that the key Uniqueness Lemma of \citeA{Antic22} may fail as well (\prettyref{w:UL}). In total, we have obtained a similarity-based framework of analogical proportions in the general setting of universal algebra with appealing mathematical properties --- most importantly, future results on algebraic similarity can be \textit{directly} applied to proportions as defined here.

\bibliographystyle{theapa}
\bibliography{/Users/christianantic/Bibdesk/Bibliography,/Users/christianantic/Bibdesk/Publications_J,/Users/christianantic/Bibdesk/Publications_C,/Users/christianantic/Bibdesk/Preprints,/Users/christianantic/Bibdesk/Submitted,/Users/christianantic/Bibdesk/Notes}
\if\isdraft1
\newpage

\section{Todos}

\todo[inline]{Finde Analogie zu FPT/UL}

\todo[inline]{$a\approx c$ und $b\approx d$ impliziert $a:b\approx c:d$ bzw. $a:b::c:d$?}

\fi
\end{document}